\documentclass[10pt,twocolumn]{IEEEtran}
\usepackage{epsfig,latexsym}
\usepackage{float}
 \usepackage{enumerate}
\usepackage{indentfirst}
\usepackage{amsmath}
\usepackage{amssymb}
\usepackage{times}
\usepackage{subfigure}
\usepackage{cite}
\usepackage{url}
\usepackage{color}

\newtheorem{Proposition}{Proposition}

\newtheorem{Lemma}{Lemma}

\newtheorem{Remark}{Remark}

\begin{document}

\title{{  NOMA: An Information Theoretic Perspective  }}
\author{Peng Xu,  Zhiguo Ding, \IEEEmembership{Member, IEEE}, Xuchu Dai and H. Vincent Poor,
 \IEEEmembership{Fellow, IEEE}
\thanks{P. Xu and X. Dai are with Dept. of Electronic Engineering and Information Science, University of Science and Technology of China,  Hefei, Anhui, China.
Z. Ding and H. V. Poor are with the Department of Electrical Engineering, Princeton
University, Princeton, NJ 08544, USA. Z.Ding is also with the School of Computing
and Communications, Lancaster University,   U.K.
}\vspace{-1em}}\maketitle
\begin{abstract}
  In this letter, the performance of non-orthogonal multiple access (NOMA)
   is investigated from an information theoretic
  perspective. {The relationships among the capacity region of  broadcast channels and two rate regions achieved by
  NOMA and time-division multiple access (TDMA) are illustrated first}.
  Then, the performance of NOMA is evaluated  by
  considering  TDMA   as the benchmark, where both the sum rate and the individual user rates are used as the criteria.
   In a wireless downlink scenario with user pairing, the developed analytical results show that
    NOMA  can outperform TDMA not only for the sum rate but also for each user's
     individual rate, particularly when the difference between the users' channels  is large.
  \end{abstract}

\vspace{-1em}
\section{Introduction}
Because of its superior spectral efficiency,  non-orthogonal multiple access (NOMA) has been recognized
as a promising technique to be used in the fifth generation (5G) networks \cite{saito2013system,li20145g,ding2014performance,timotheou2015fairness}.
 NOMA  utilizes
the power domain for achieving multiple access, i.e., different users are served at different power levels. Unlike conventional orthogonal MA, such as time-division multiple access (TDMA),
 NOMA  faces strong   co-channel interference between different users, and successive interference cancellation (SIC) is used by the NOMA users with better channel conditions for interference management.

The concept of NOMA is essentially
 a special case of superposition coding developed for broadcast channels (BC).
Cover first found  the capacity region of a degraded  discrete memoryless  BC by using   superposition coding \cite{cover1972broadcast}.
Then, the capacity region of the Gaussian BC  with single-antenna terminals  was established in  \cite{bergmans1974simple}.
Moreover, the capacity region of the multiple-input multiple-output (MIMO) Gaussian BC was found in \cite{weingarten2006capacity},
by applying dirty paper coding (DPC) instead of  superposition coding.
This paper mainly focuses on the single-antenna scenario.

Specifically,
consider a Gaussian BC with a single-antenna  transmitter and two single-antenna
 receivers, where each receiver is corrupted by additive  Gaussian
 noise with unit variance.
Denote the ordered  channel gains from the transmitter  to the two receivers by $h_w$ and $h_b$, i.e.,
 $|h_w|^2<|h_b|^2$. For a given channel pair $(h_w,h_b)$, the capacity region is given by \cite{bergmans1974simple}
\begin{align}
\mathcal{C}^{\textrm{BC}}\triangleq &\bigcup_{\begin{footnotesize}\begin{array}{l}a_1\!+\!a_2\!=\!1,a_1,a_2\geq0
\end{array}\end{footnotesize}} \Bigg\{(R_1,R_2):    R_1,R_2\geq 0,\nonumber\\
 &R_1\!\leq \!\log_2\left(1\!+\!\frac{a_1x}{1\!+\!a_2x}\right), R_2\!\leq\!\log_2\left({1\!+\!a_2y}\right)\Bigg\}\label{rate_BC},
\end{align}
where $a_i$ denotes the power allocation coefficient,   $x=|h_w|^2\rho$, $y=|h_b|^2\rho$, and $\rho$ denotes the transmit signal-noise-ratio (SNR). Based on SIC, the rate region achieved by NOMA, denoted by $\mathcal{R}^{\textrm{N}}$, can be expressed   the same as $\mathcal{C}^{\textrm{BC}}$ in \eqref{rate_BC}, but
with an additional  constraint $a_1\geq a_2$  in order to guarantee the quality of service at the user with
the poorer  channel condition.

In addition, the TDMA region   is given by
\begin{align}
\mathcal{R}^{\textrm{T}}\triangleq  \Bigg\{(R_1,R_2): & R_1,R_2\geq 0,\frac{R_1}{R_1^*}+\frac{R_2}{R_2^*}\leq 1\Bigg\},
\end{align}
where $R_1^*=\log_2 (1+x)$ and $R_2^*=\log_2 (1+y)$.

The three regions are illustrated  in Fig. \ref{regions_comparison}.
In the rest of this letter, we are interested in the two region boundaries, i.e., the curve A-F and the segment A-E,
which represent the optimal  rate pairs achieved by NOMA and TDMA, respectively.
The relationship between  the rate pairs  on the   two boundaries will be further
 interpreted based on plane geometry as shown in Section \ref{comparison}.
  Then, based on their relationship, the performance of NOMA is
    characterized  in terms of both the sum rate and individual user rates, by considering the conventional TDMA scheme
  as the benchmark.   In a wireless downlink scenario  with user pairing,
  analytical results are developed to demonstrate that NOMA can outperform TDMA when
  there exists a significant difference between the channel conditions  of the scheduled   users.

\begin{figure}[tbp]\centering
    \epsfig{file=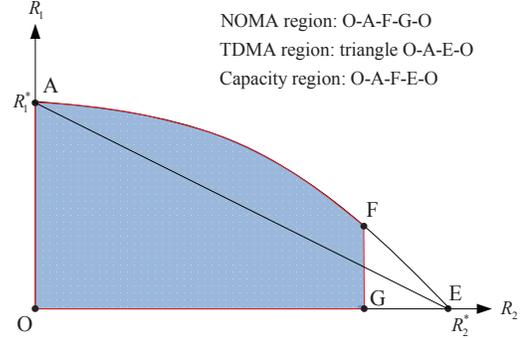, width=0.4\textwidth,clip=}
\caption{The capacity region, NOMA  and TDMA regions for a given channel pair $(h_w,h_b)$,
where the
point F is located at $\left(\log_2(1\!+\!\frac{y}{2}),\log_2\left(1\!+\!\frac{x}{2\!+\!x}\right)\right)$.}\label{regions_comparison}
\end{figure}

\section{Preliminary}
Two propositions are provided in this section, which will be used in the next section. Specifically,
define $f^N(\cdot)$  and $f^T(\cdot)$ as the following two functions:
\begin{align}
 f^N(z)&= \log_2\left(\frac{(1+x)y}{y+(2^{z}-1)x}\right),\ 0\leq z \leq R_2^*,\label{fNz}\\
f^T(z)&=\left(1-\frac{z}{R_2^*}\right){R_1^*},\ 0\leq z \leq R_2^*.\label{fTz}
\end{align}

For a given $z_{0}\in(0,R_2^*)$,  two propositions are provided as follows.
\begin{Proposition}
  If $z>z_0$, then $f^N(z)+z>f^T(z_0)+z_0$.
\end{Proposition}
\begin{proof}
  We can first obtain $\frac{d\; f^N}{d\; z}=\frac{-x2^{z}}{y-x+x2^{z}}$.
  Define $f_{sum}^N=f^N+z$, so $\frac{d\;f_{sum}^N}{d\; z}=
  \frac{y-x}{y-x+x2^{z}}>0$, which means that $f_{sum}^N(z)$ is a monotonic increasing function of $z$.
  Thus,  $f_{sum}^N(z)>f_{sum}^N(z_0)=f^N(z_0)+z_0$.

  On the other hand,  $f^N(z)$ is a  concave  function of $z$ (i.e., $\frac{d^2\; f^N}{d\; z^2}<0$)
   when $z\in(0,R_2^*)$. Hence
   $$\lambda f^N(0)+(1-\lambda)f^N(R_2^*)\leq f^N(\lambda\times 0+(1-\lambda)R_2^*),$$
 for $\forall \lambda\in(0,1)$. Since $f^N(0)=R_1^*$ and $f^N(R_2^*)=0$, we can obtain
   $f^N(z_0)\geq f^T  (z_0)$ by setting $\lambda=1-z_0/R_2^*$.
 This proposition has been proved.
\end{proof}
\begin{Proposition}
  If $z<z_0$, then $f^N(z)>f^T(z_0)$.
\end{Proposition}
\begin{proof}
  Since  $f^N(z)$ is monotonically  decreasing in $z$, $f^N(z)>
  f^N(z_0)$ for $z<z_0$. Furthermore, due to the fact that $f^N(z)$ is a concave function of $z$,   $f^N(z_0)>f^T (z_0)$ as discussed above.
This proposition has been proved.
\end{proof}

\section{Performance Analysis}
In this section, the performance of NOMA will be studied  by considering the achievable rates
of TDMA as a benchmark.
\vspace{-1em}
\subsection{Comparison to TDMA}\label{comparison}
Here, the individual rates and the sum rate achieved by NOMA will be compared with those of TDMA
using plane geometry.

As shown in Fig. \ref{regions}, for a given channel pair $(h_w,h_b)$ and $|h_w|<|h_b|$,  suppose  the point N is located at
 \begin{align}\label{R12N}
(R_2^N,R_1^N)=\left(\log_2(1+a_2y),\log_2\left(1+\frac{a_1x}{1+a_2x}\right)\right),\end{align}
where $a_1\!+\!a_2\!=1$, $0\!\leq a_2\!\leq a_1$;  and the  point T is located at
\begin{align}(R_2^T,R_1^T)=\left(b_2R_2^*,b_1R_1^*\right),\end{align}
 where $b_1\!+\!b_2=\!1$, $b_1,b_2\geq 0$.
This means that the points N and T lie on the curve A-F (NOMA rate pair) and the segment A-E
(TDMA rate pair), respectively.
 In addition, consider three important  lines: $R_1=\!R_1^N$, $R_2=R_2^N$ and $R_1\!+R_2=\!R_1^N\!+\!R_2^N$,
 which represent the two NOMA users' individual rates and their sum rate, respectively. It is easy to prove that
 $R_1^N+R_2^N<R_2^*$, and these three lines will
  divide the line segment A-E into four subsegments with intersection points B, C and D.

 When considering  $h_w$ and $h_b$ to be random variables and fixing $a_i$,  we can define four  random events
 according to the location of point T and  these  four subsegments as follows.
 \begin{align}
   \varepsilon_1&\triangleq\left\{\textrm{Point T lies on
   subsegment A-B}\right\},\\
    \varepsilon_2&\triangleq\left\{\textrm{Point T lies on
   subsegment B-C}\right\},\\
      \varepsilon_3&\triangleq\left\{\textrm{Point T lies on
   subsegment C-D}\right\},\\
      \varepsilon_4&\triangleq\left\{\textrm{Point T lies on
   subsegment D-E}\right\} .
 \end{align}
 These events comprehensively reflect the relationship between the rates (including the individual rates and the
 sum rate) of NOMA and TDMA, i.e.,
\begin{align}
   \varepsilon_1&\! =\!  \left\{R_1^N<R_1^T, R_2^N>R_2^T,  R_1^N\! +\! R_2^N>R_1^T\!  +\! R_2^T\right\},\\
    \varepsilon_2&\! =\!  \left\{R_1^N>R_1^T, R_2^N>R_2^T,  R_1^N\! +\! R_2^N>R_1^T \! +\! R_2^T\right\},
    \end{align}
   \begin{align}
      \varepsilon_3&\! =\!  \left\{R_1^N>R_1^T, R_2^N<R_2^T,  R_1^N\! +\! R_2^N>R_1^T \! +\! R_2^T\right\},\\
       \varepsilon_4&\!=\! \left\{R_1^N>R_1^T, R_2^N<R_2^T,  R_1^N\!+\!R_2^N<R_1^T \!+\!R_2^T\right\}.
 \end{align}
 \begin{figure}[tbp]\centering
  \epsfig{file=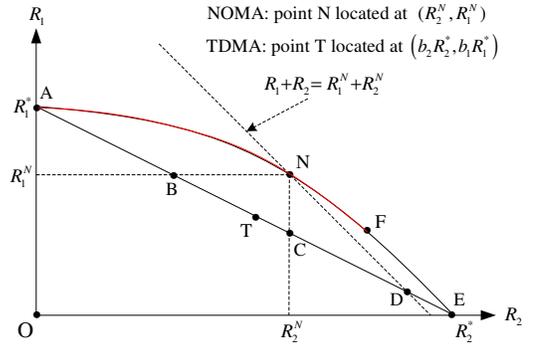, width=0.415\textwidth,clip=}
\caption{Comparison of the rate pairs achieved  by NOMA and TDMA  schemes for a given channel pair $(h_w,h_b)$,
where $(R_2^N,R_1^N)$ is defined in \eqref{R12N}.
 In this example, $T$ lies on the line segment B-C.}\label{regions}
\end{figure}
 Notice that $(R_1^N,R_2^N,R_1^T,R_2^T)$ satisfies the relationship $R_1^N=f^N(R_2^N)$ and
 $R_1^T=f^T(R_2^T)$ as shown in \eqref{fNz} and \eqref{fTz}. Hence,
 based on Propositions 1 and 2  by replacing $(f^N,z,f^T,z_0)$ with
 $(R_1^N,R_2^N,R_1^T,R_2^T)$,
 we can remove some redundant
 conditions for each event, i.e.,
 \begin{align}
   \varepsilon_1 &\stackrel{(a)}{=} \left\{R_1^N<R_1^T, R_2^N>R_2^T\right\},\\
    \varepsilon_2&\stackrel{(b)}{=} \left\{R_1^N>R_1^T, R_2^N>R_2^T\right\},\label{event2}\\
      \varepsilon_3&\stackrel{(c)}{=} \left\{ R_2^N<R_2^T,  R_1^N+R_2^N>R_1^T +R_2^T\right\},\\
      \varepsilon_4&\stackrel{(d)}{=} \left\{R_1^N+R_2^N<R_1^T +R_2^T\right\},
 \end{align}
 where $(a)$ and $(b)$ are based on Proposition 1; $(c)$ is based on Proposition 2; and
 $(d)$ is based on the converse-negative proposition of Proposition 1 (i.e.,
 $\{R_1^N\!+\!R_2^N<R_1^T \!+\!R_2^T\}\Rightarrow \{R_2^N<R_2^T\}$) and  Proposition 2.
\begin{Remark}
  Among these four events, of particular interest is $\varepsilon_2$ which represents the
  situation in which  NOMA outperforms TDMA
  in terms of not only the sum rate but also each individual rate.
\end{Remark}

 In the next subsection,  the probability of each event
 will be calculated, which characterizes the    performance of NOMA in comparison with TDMA.

\vspace{-0.8em}
\subsection{Probability Analysis}\label{Probability_ana}
Let the two users in the considered BC be selected from $M$ mobile users in a downlink
 communication scenario, as motivated  in \cite{ding2014impact}.
 Without loss of generality, assume that all the users' channels are ordered as
 $|h_1|^2\leq \cdots \leq |h_M|^2$, where $h_m$ is the Rayleigh fading channel gain from the base station
 to the $m$-th user. Considered that the $m$-th user is paired with the $n$-th user to perform NOMA.
 Hence $h_w=h_m$, $h_b=h_n$, and $x$ and $y$ can be rewritten as
  $x=\rho |h_m|^2$, $y=\rho |h_m|^2$, with joint probability density function (PDF)  as follows
   \cite{david2003order}:
\begin{align}
  f_{X,Y}(x,y)=&w_1f(x)f(y)[F(x)]^{m-1}[1-F(y)]^{M-n}\nonumber\\
  &\times [F(y)-F(x)]^{n-1-m}, \ 0<x<y,
\end{align}
where $f(x)=\frac{1}{\rho}e^{-\frac{x}{\rho}}$, $F(x)=1-e^{-\frac{x}{\rho}}$, and
$w_1=\frac{M!}{(m-1)!(n-1-m)!(M-n)!}$.
A fixed power allocation strategy   $(a_1,a_2)$ is considered in this NOMA system  for the sake of simplicity.
Dynamically changing  $(a_1,a_2)$
according to the random channel state information  (CSI) could achieve a larger ergodic rate region \cite{li2001capacity}, but at the expense of higher complexity.

Using the PDF of $x$ and $y$, the probability of  each event defined in the previous subsection can be calculated
in order to evaluate the performance of NOMA.
 The probability of the event $\varepsilon_2$
 is first given in the following lemma, where we set $b_2=1/2$ (each user is allocated an equal-length time slot,
 which is also called ``naive TDMA'' in \cite{cover2006elements}) for simplicity.
 \begin{Lemma}
 Given $(M,m,n,\rho,a_2)$ and $b_2=1/2$, the probability  that  NOMA achieves
 larger  individual rates than conventional TDMA  for both user $m$ and user $n$ is given by
 \begin{align}
  P(\varepsilon_2)\!=&w_1\sum_{k=0}^{m-1}(-1)^{m-\!1-k}C_{m-1}^k
           \Bigg[\sum_{i=0}^{n-1}\frac{(-1)^{n-1-i}C_{n-\!1}^id^{M\!-i}}{M-i}\nonumber\\
           &- \sum_{i=0}^{k}\sum_{j=0}^{n-1-k}\frac{(-1)^{n-1-i-j}C_{k}^iC_{n-1-k}^jd^{M-i}}{M-i-j}\Bigg].
           \end{align}
   \end{Lemma}
\begin{proof}
From \eqref{event2}, $P(\varepsilon_2)$ can be calculated as follows:
 \begin{align}\label{P2}
  P(&\varepsilon_2)=P(R_1^N>R_1^*/2, R_2^N>R_2^*/2)\nonumber\\
  &\stackrel{(a)}{=}P\left(x<w_2,y>w_2\right)\nonumber\\
  &=w_1\int_{w_2}^{+\infty}f(y)[1-F(y)]^{M-n}\nonumber\\
  &\hspace{1em}\times\left(\int_{0}^{w_2}f(x)[F(x)]^{m-1}[F(y)-F(x)]^{n-1-m} dx\right)dy\nonumber\\
    &\stackrel{(b)}{=}w_1\int_{w_2}^{+\infty}f(y)[1-F(y)]^{M-n}
    \nonumber\\
  &\hspace{1em}\times\left(\int_{d}^{1}(1-t)^{m-1}(t-e^{-\frac{y}{\rho}})^{n-1-m}
     dt\right)dy\nonumber\\
        &{=}w_1\int_{w_2}^{+\infty}f(y)[1-F(y)]^{M-n}\nonumber\\
  &\hspace{1em}\times\left( \int_{d-e^{-\frac{y}{\rho}}}^{1-e^{-\frac{y}{\rho}}}
       (1-e^{-\frac{y}{\rho}}-u)^{m-1} u^{n-1-m}     du\right)dy\nonumber\\
         &\stackrel{(c)}{=}w_1\int_{w_2}^{+\infty}f(y)[1-F(y)]^{M-n}
         \sum_{k=0}^{m-1}(-1)^{m-1-k}C_{m-1}^k\nonumber\\
  &\hspace{1em}\times
         \frac{(1-e^{-\frac{y}{\rho}})^{n-1}-(1-e^{-\frac{y}{\rho}})^{k}
         (d-e^{-\frac{y}{\rho}})^{n-1-k}}{n-1-k}dy\nonumber\\
           &{=}w_1\sum_{k=0}^{m-1} \frac{(-1)^{m-1-k}C_{m-1}^k}{n-1-k}
           \Bigg[\underbrace{\int_{0}^{d}v^{M-n}(1-v)^{n-1}dv}
           _{Q_1}\nonumber\\
  &\hspace{3em}-\underbrace{\int_{0}^{d}v^{M-n}(1-v)^{k}(d-v)^{n-1-k}dv}
           _{Q_{2,k}}\Bigg]
\end{align}
where $(a)$ follows the definition  $w_2=\frac{1-2a_2}{a_2^2}$; $(b)$ follows $d=e^{-\frac{w_2}{\rho}}$; and
$(c)$ follows $C_p^q=\frac{p!}{q!(p-q)!}$, $p>q$. 
Furthermore, the two terms $Q_1$ and $Q_{2,k}$ can be calculated as follows:
\begin{align}
  Q_1&=\sum_{i=0}^{n-1}(-1)^{n-1-i}C_{n-1}^i \int_{0}^dv^{M-1-i}dv\nonumber\\
 & = \sum_{i=0}^{n-1}\frac{(-1)^{n-1-i}C_{n-1}^id^{M-i}}{M-i}, \label{Q1}\\
  Q_{2,k}&\!=\!\sum_{i=0}^{k}\sum_{j=0}^{n-1-k}\!(-1)^{n\!-1-i-\!j}C_{k}^iC_{n\!-1-k}^j d^j\int_{0}^dv^{M\!-1-i-\!j}dv\nonumber\\
  &= \sum_{i=0}^{k}\sum_{j=0}^{n-1-k}\frac{(-1)^{n-1-i-j}C_{k}^iC_{n-1-k}^jd^{M-i}}{M-i-j}.\label{Q2k}
\end{align}
Substituting the above two relationships into \eqref{P2}, this lemma has been proved.
\end{proof}

 Moreover, for the first event, it is not difficult to obtain that
    \begin{align}\label{P1}
   &P(\varepsilon_1)= P(R_2^N>R_2^T)-P(\varepsilon_2)\nonumber\\
   &=1\!-\!w_3\sum_{i=0}^{n-1}
   \frac{(-1)^iC_{n-1}^i}{M\!-n+i\!+1}\left(1\!-d^{M\!-n+i+\!1}\right)\!-\!P(\varepsilon_2),
  \end{align}
  where $w_3=\frac{M!}{(m-1)!(M-m)!}$. For the fourth event,
  from \cite{ding2014impact} (Theorem 1), we have
  \begin{align}\label{P4}
   & P(\varepsilon_4)=P(R_1^N+R_2^N<R_1^T+R_2^T)=\nonumber\\
   &1-\!w_1\!\sum_{i=0}^{n-1-m}\frac{(-1)^iC_{n-1-m}^i}{m+i}\int_{\sqrt{w_2\!+1}\!-1}^{w_2}f(y)[F(y)]^{n\!-1-m-i}\nonumber\\
   &\ \times[1-F(y)]^{M-n}\left([F(y)]^{m+i}-\left[F\left(\frac{w_2-y}{1+y}\right)\right]^{m+i}\right)dy                         \nonumber\\
   &\qquad-{w_3}\sum_{j=0}^{n-1}\frac{(-1)^jC_{n-1}^j}{M-n+j+1}d^{-{(M-n+j+1)}}.
  \end{align}
  Thus, the probability of the third event can be written as
  \begin{align}\label{P3}
    P(\varepsilon_3)=1-P(\varepsilon_1)-P(\varepsilon_2)-P(\varepsilon_4).
  \end{align}

  Now, $P(\varepsilon_i)$, $i=2,1,4,3$, have been obtained as in Eqs. \eqref{P2}, \eqref{P1}, \eqref{P4},
  and \eqref{P3}, respectively.

  {\bf Special Case:}  The expression for $P(\varepsilon_2)$ in Lemma 1 can be simplified when
   considering a special pairing  case, i.e., $m=1$, $n=M$. In this case, $k=0$, and $Q_1$ and $Q_{2,0}$
   in \eqref{Q1} and \eqref{Q2k}, respectively, can be derived as
  \begin{align}
   & Q_1= -\frac{1}{M}\sum_{i=0}^{M-1}{(-1)^{M-i}C_{M}^id^{M-i}}\nonumber\\
   & =\!-\frac{1}{M}\left[\sum_{i=0}^{M}{C_{M}^i(-d)^{M-i}}-1\right]\!=\!\frac{1-\!(1-\!d)^M}{M},
  \end{align}
  \begin{align}
  Q_{2,0}&\!=\!  \sum_{j=0}^{M-1}\frac{(-1)^{M\!-1-\!j}C_{M\!-1}^jd^{M}}{M-j}
  \!=\! -\frac{d^M}{M}\sum_{i=0}^{M-1} (-1)^{M\!-j}C_M^j\nonumber\\
  &=-\frac{d^M}{M}\left[\sum_{i=0}^{M} (-1)^{M-j}C_M^j-1\right]\nonumber\\
  &=-\frac{d^M}{M}\left[(1-1)^M-1\right]=\frac{d^M}{M}.
  \end{align}
 Hence, $ P(\varepsilon_2)=1-(1-d)^M-d^M$. The optimal $d$ in this case is $1/2$,
  which implies $a_2=\frac{-1+\sqrt{1+\rho \ln2}}{\rho \ln2}<\frac{1}{2}$, and
  $ P(\varepsilon_2)=1-\frac{1}{2^{M-1}}$. Here $\ln(\cdot)$ denotes the natural logarithm.
  \begin{Remark}
    This special case shows that $ P(\varepsilon_2)\rightarrow 1$ when  $M$ is sufficiently large.
  This means that, almost for all the possible
 channel realizations, NOMA achieves larger individual rates than naive  TDMA
    for both user $m$ and $n$ as long as
    the difference between the better and worse channel gains is sufficiently large. This phenomenon is
    also valid for other pairing  cases (i.e., $(n,m)\neq (M,1)$) as verified via some numerical examples
     in the next section.
  \end{Remark}
  \section{Numerical Results}
  \begin{figure}[tbp]\centering
    \epsfig{file=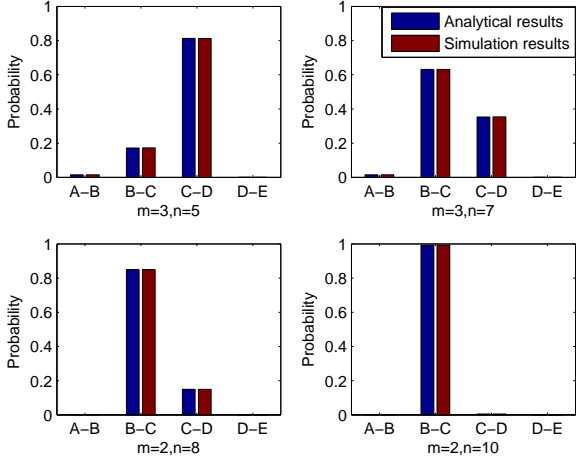, width=0.49\textwidth,clip=}
\caption{Probability that the point T lies on a certain segment in Fig. \ref{regions}, where $\rho=25$ (dB), $a_2=1/{\sqrt{\rho}}$.}
\label{Each_probability}
\end{figure}
In this section, the performance of NOMA is evaluated in comparison with     TDMA by
using computer simulations. The total number of users in the wireless downlink system is $M=10$, and different choices of $(m,n)$ will be considered.

In Fig. \ref{Each_probability}, the probability of each event defined in Section
\ref{comparison} is displayed via column diagrams. Specifically, the probabilities that
the point T lies on subsegments A-B, B-C, C-D and D-E in Fig. \ref{regions}
 are displayed, where we set $\rho=25$ dB and
  $a_2=1/{\sqrt{\rho}}$  for simplicity.
 Four different user pairs $(m,n)$ are considered, which shows that the probability  that the
  point T lies on subsegment B-C (i.e.,
 $P(\varepsilon_2)$)
 increases with the value of $(n-m)$, as discussed in Remark 2.
 In Fig. \ref{P_BC}, additional  numerical results  are provided to show $P(\varepsilon_2)$ as a function of  $n$.
  As shown in this figure,
 $P(\varepsilon_2)$ increases with the value of $(n-m)$,  i.e., NOMA is prone to perform better than TDMA in terms of
     individual rates when    the difference between the users' channels  is large.
  In addition, it is worth pointing out that the Monte Carlo  simulation results provided in Figs. \ref{Each_probability}
  and \ref{P_BC} match well with the analytical results developed in \eqref{P2},  \eqref{P1},  \eqref{P4} and \eqref{P3}.
   In Fig. \ref{Average_rate}, individual rates of NOMA and TDMA averaged over the fading channels
    are depicted as functions of SNR (i.e., $\rho$),
  where we set $(m,n)=(1,M)$, and $a_2={(\sqrt{\rho\ln(2)+1}-1)}/{{(\rho\ln(2))}}$ according to
  the special case in Section \ref{Probability_ana}.  As shown in this figure, NOMA has a constant performance gain
  over TDMA for each individual rate. When $\rho=55$ dB,
   the performance gains with respect to user $m$ and user $n$ are about
  1 bits per channel uses (BPCU) and 2 BPCU, respectively. This is due to the fact that
 $P(\varepsilon_2)\rightarrow 1$ in this case, i.e., NOMA outperforms TDMA in terms of each user's rate for
  almost all the possible
 realizations of $(h_m,h_n)$.
   \begin{figure}[tbp]\centering
    \epsfig{file=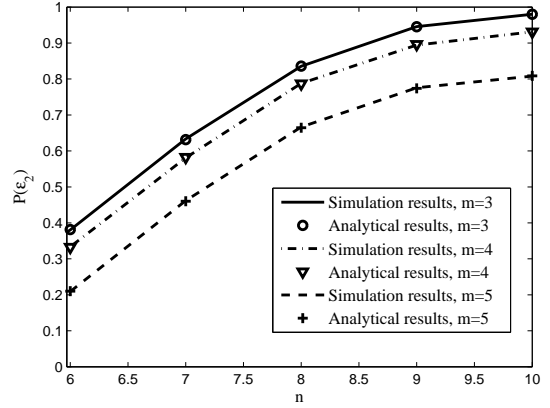, width=0.44\textwidth,clip=}
\caption{Probability of the event $\varepsilon_2$, where $\rho=25$ (dB), $a_2=1/{\sqrt{\rho}}$.}\label{P_BC}
\end{figure}
  \begin{figure}[tbp]\centering
    \epsfig{file=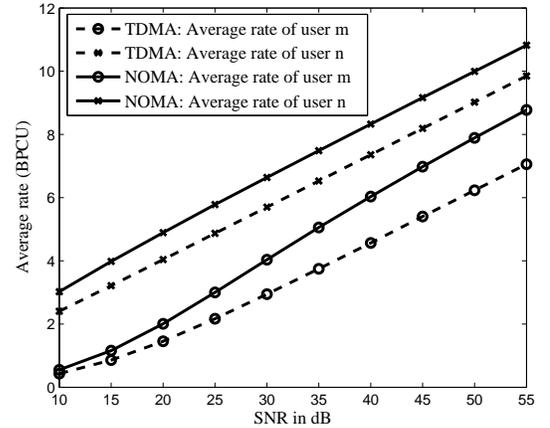, width=0.455\textwidth,clip=}
\caption{Average individual rates, $a_2=\frac{\sqrt{\rho\ln(2)+1}-1}{{\rho\ln(2)}}$, $m=1$, $n=10$.}\label{Average_rate}
\end{figure}
\vspace{-1em}
\section{Conclusions}
This letter has investigated the performance of NOMA in a downlink network
from an information theoretic  perspective. The relationship among  the  BC  capacity region,
the NOMA rate region and the TDMA rate region was first described.
According to their relationship, the performance of NOMA was   evaluated
 in terms of both the sum rate and
 users' individual rate, by
  considering   TDMA
  as the benchmark.
  Future work of interest is to dynamically change power allocation   according to instantaneous CSI
  for enlarging the ergodic  achievable rates \cite{li2001capacity}. Moreover, it is important to establish the connection between MIMO NOMA and information theoretic MIMO broadcasting concepts.
\section{Acknowledgment}
The authors thank Mr Yiran Xu for helpful discussions. 
\newpage
\bibliographystyle{IEEEtran}
\bibliography{references}
\end{document}